\documentclass[conference]{IEEEtran}
\usepackage{amsmath}
\usepackage{amsthm}
\usepackage{array}
\usepackage{mdwmath}
\usepackage{mdwtab}
\usepackage{eqparbox}
\usepackage{amsfonts}
\usepackage{tikz}
\usepackage{amsthm}
\usepackage{array}
\usepackage{mdwmath}
\usepackage{mdwtab}
\usepackage{eqparbox}
\usepackage{tikz}
\usepackage{latexsym}
\usepackage{cite}
\usepackage{amssymb}
\usepackage{bm,bbm}
\usepackage{amssymb}
\usepackage{graphicx}
\usepackage{mathrsfs}
\usepackage{epsfig}
\usepackage{psfrag}
\usepackage{setspace}
\usepackage{hyperref}
\usepackage{algorithm}
\usepackage{algpseudocode}
\usepackage{stfloats}

\newtheorem{theorem}{Theorem}

\newtheorem{lemma}{Lemma} 
\newtheorem{corollary}{Corollary}

\newtheorem{definition}{Definition}

\def \bE {\mathbb{E}}

\def \bR {\mathbb{R}}

\def \calX {\mathcal{X}}
\DeclareMathOperator*{\argmax}{arg\,max}

\DeclareMathOperator*{\esssup}{ess\,sup}

\title{Dependence Measures Bounding the Exploration Bias for General Measurements}

\author{

\IEEEauthorblockN{Jiantao Jiao}
\IEEEauthorblockA{
EE Department, Stanford University\\
jiantao@stanford.edu}

\and

\IEEEauthorblockN{Yanjun Han}
\IEEEauthorblockA{EE Department, Stanford University\\
yjhan@stanford.edu
}

\and
\IEEEauthorblockN{Tsachy Weissman}
\IEEEauthorblockA{EE Department, Stanford University\\
tsachy@stanford.edu}
\thanks{This work was supported in part by the NSF Center for Science of Information under grant agreement CCF-0939370.}
}

\begin{document}
\maketitle

\begin{abstract}
We propose a framework to analyze and quantify the bias in adaptive data analysis. It generalizes that proposed by Russo and Zou'15, applying to measurements whose moment generating function exists, measurements with a finite $p$-norm, and measurements in general Orlicz spaces. We introduce a new class of dependence measures which retain key properties of mutual information while more effectively quantifying the exploration bias for heavy tailed distributions. We provide examples of cases where our bounds are nearly tight in situations where the original framework of Russo and Zou'15 does not apply. 
\end{abstract}


\section{Introduction}

Suppose we have $n$ measurements $\phi_i, 1\leq i\leq n$ of a dataset, and wish to select one of the measurements for further processing. Settings of this flavor appear frequently in applications such as model selection and reinforcement learning, where the statistician wants to exploit the information collected to infer the ground truth. We denote the expectations of each measurement $\phi_i$ as $\mu_i$, i.e., $\mathbb{E}\phi_i = \mu_i$. We also denote the index of the measurement selected as random variable $T\in \{1,2,\ldots,n\}$. The natural question is, how much does $\mathbb{E}\phi_T$ deviate from $\mathbb{E}\mu_T$? This question is of practical importance, since a deviation of $\mathbb{E}\phi_T$ from $\mathbb{E}\mu_T$ corresponds to a misguided rule or generalization error for selecting from the components of $\bm{\phi} = (\phi_1,\phi_2,\ldots,\phi_n)$. 

It was shown in Russo and Zou~\cite{Russo--Zou2015much} that one can bound the bias of the data exploration, i.e., the quantity $\mathbb{E}[\phi_T - \mu_T]$ as follows:
\begin{theorem}\cite[Prop. 3.1.]{Russo--Zou2015much}\label{thm.russozou}
If $\phi_i - \mu_i$ is $\sigma$-sub-Gaussian for each $i\in \{1,2,\ldots,n\}$, then
\begin{align}\label{eqn.russozoubound}
|\mathbb{E}[\phi_T - \mu_T]| \leq \sigma \sqrt{2 I(T; \bm{\phi})},
\end{align}
where $I(T;\bm{\phi})$ denotes the mutual information between $T$ and $\bm{\phi}$. 
\end{theorem}

Moreover, Russo and Zou~\cite{Russo--Zou2015much} argued that for certain selection rules $T$ that are variants of selecting the maximum among $\{\phi_1,\phi_2,\ldots,\phi_n\}$, this bound is tight for Gaussian and exponential distributions. Indeed, if $\phi_i \stackrel{\mathrm{i.i.d.}}{\sim} \mathcal{N}(0,\sigma^2)$ and $T = \argmax_i \{\phi_1,\phi_2,\ldots,\phi_n\}$, it is well known that
\begin{align} 
\frac{\phi_T}{\sigma\sqrt{2\ln n}} \stackrel{\mathrm{a.s.}}{\rightarrow} 1\quad \textrm{ as }n \rightarrow \infty,
\end{align}
and $I(T; \bm{\phi}) = H(T) = \ln n$. 

The interested readers are referred to Russo and Zou~\cite{Russo--Zou2015much} for variations on this bound and its applications. Our work is motivated by the observations that 

\begin{enumerate}
\item Theorem~\ref{thm.russozou} assumes sub-Gaussian distributions\footnote{A generalization to sub-exponential distributions, which does not follow the form of the inequality in Theorem~\ref{thm.russozou}, was also presented in~\cite{Russo--Zou2015much}. It is tightened by Corollary~\ref{cor.subexponential} of this paper. }, whereas in many real world applications, such as natural language processing and e-commerce recommendation systems, the measurements follow long-tail distributions which are neither sub-Gaussian nor sub-exponential;

\item Lower bounds corresponding to Theorem~\ref{thm.russozou} were proved for some specific selection rules and assuming Gaussian distributions in~\cite{Russo--Zou2015much}. 
\end{enumerate}
 
In this context, our main contributions are the following:

\begin{itemize}
\item We generalize Theorem~\ref{thm.russozou} to all distributions with a non-trivial moment generating function. We show that for such distributions, the bound on the right hand side of (\ref{eqn.russozoubound}) is replaced by a function $f(I(T; \bm{\phi}))$. For Gaussian random variables the function specializes to $f(x) = \sigma\sqrt{2x}$. 

\item We introduce a new class of dependence measures $I_\alpha(X;Y)$ paralleling mutual information. Concretely, for $1\leq \alpha<\infty$, we define
\begin{align}
I_\alpha(X;Y) = D_{\phi_\alpha}(P_{XY}\|P_X P_Y),
\end{align}
where $D_{\phi_\alpha}(P\|Q)$ is the $\phi$-divergence generated by the convex function $
\phi_\alpha(x)= |x-1|^\alpha$. Clearly $I_\alpha(X;Y)\geq 0$ and $I_\alpha(X;Y) = 0$ if and only if $X$ and $Y$ are independent. It satisfies the following.
\begin{lemma}\label{lemma.boundalphamutual}
Suppose $X$ takes values in a finite alphabet with cardinality $|\mathcal{X}|$, while $Y$ is arbitrary. Then
\begin{align}
I_\alpha(X;Y) \leq 1 + \sum_{x\in\calX} [P_X(x)]^2\left(\left|\frac{1}{P_X(x)}-1\right|^\alpha-1\right)
\end{align}
which is tight iff $X$ is a function of $Y$. In particular,
\begin{align}
I_\alpha(X;Y)\leq \frac{|\calX|-1}{|\calX|}\left[(|\calX|-1)^{\alpha-1}+1\right] < 1+|\calX|^{\alpha-1}
\end{align}
for $1\leq \alpha\leq 2$.
\end{lemma}
We show that these measures arise in bounding the exploration bias for distributions whose moment generating functions do not exist. We present theorems paralleling Theorem~\ref{thm.russozou} for such heavy tailed distributions, and construct examples implying that our bounds are essentially tight for a sequence of non-Gaussian heavy tailed distributions. Our results imply that mutual information is not fundamental to bounding exploration bias, and one should apply different functionals for distributions with different tail behaviors. We conclude with connections to the literature of maximal inequalities and a generalization to random variables in any Orlicz spaces.  
\end{itemize}



\section{Preliminaries}\label{sec.preliminaries}

The cumulant generating function of a random variable $X$ is defined as 
\begin{align}
\psi(\lambda) = \ln \mathbb{E} e^{\lambda X}, \lambda \geq 0. 
\end{align}
We assume that there exists a $\lambda>0$ such that $\mathbb{E}e^{\lambda X}<\infty$. It follows from H\"{o}lder's inequality that there exists an interval $(0,b), 0<b\leq \infty$ such that $\psi(\lambda)<\infty$ for all $\lambda \in (0,b)$, and $\psi(\lambda)$ is convex on this interval. 

A random variable is called $\sigma$-sub-Gaussian if $\mathbb{E}e^{\lambda X}\leq e^{\frac{\lambda^2 \sigma^2}{2}}$ for all $\lambda \in \mathbb{R}$. A random variable is called sub-exponential with parameter $(\sigma,b)$ if $\mathbb{E}e^{\lambda X}\leq e^{\frac{\lambda^2 \sigma^2}{2}}$ for all $0\leq \lambda < \frac{1}{b}$. A random variable is called sub-gamma on the right tail with variance factor $\sigma^2$ and scale parameter $c$, i.e. $\Gamma_+(\sigma^2,c)$ if 
\begin{align}
\psi(\lambda) \leq \frac{\lambda^2 \sigma^2}{2(1-c\lambda)}\quad \textrm{ for all }0<\lambda<\frac{1}{c}. 
\end{align}
We note that the $\chi^2$ distribution with $p$ degrees of freedom belongs to $\Gamma_+(2p,2)$. 

The $\beta$-norm of a random variable $X$ for $\beta\geq 1$ is defined as
\begin{align}
\| X \|_\beta = \begin{cases}  (\mathbb{E} |X|^\beta)^{1/\beta} & 1\leq \beta <\infty \\
\esssup |X| & \beta = \infty \end{cases},
\end{align}
where the essential supremum is defined as 
\begin{align}
\esssup X = \inf\{M: \mathbb{P}(X>M) = 0\}. 
\end{align}

The Fenchel--Young inequality states that for any function $f$ and its convex conjugate $f^*$, we have
\begin{align}
f(x) + f^*(y) \geq \langle x, y \rangle, \textrm{ for all }x\in X, y\in X^*,
\end{align}
which follows from the definition of convex conjugate $f^*(y) = \sup_{x\in X} \{ \langle x, y\rangle - f(x) \}$. It follows from the Fenchel--Moreau theorem that $f = f^{**}$ if and only if $f$ is convex and lower semi-continuous. 

Csisz{\'a}r\cite{Csisz1967}, and independently Ali and Silvey\cite{Ali--Silvey1966}, introduced $\phi$-divergences defined as follows:
\begin{definition}[$\phi$-divergence]
The general form of $\phi$-divergences is
\begin{align}
D_\phi(P\|Q) = \int \phi\left(\frac{dP}{dQ} \right) dQ,
\end{align}
where $\phi: \mathbb{R}_{\geq 0}\mapsto \mathbb{R}$ is a convex and lower semi-continuous function and satisfies $\phi(1) = 0$. 
\end{definition}
It is clear that $D_\phi(P\|Q) = D_{\mathrm{KL}}(P\|Q)$ when $\phi(x) = x\ln x -x+1$. 

For two nonnegative sequences $\{a_n\}$ and $\{b_n\}$, we say $a_n \lesssim b_n$ if there exists a constant $C>0$ such that $\limsup_n \frac{a_n}{b_n}\leq C$. We say $a_n \gtrsim b_n$ if $b_n \lesssim a_n$.

\section{Main results}\label{sec.mainresults}

\begin{theorem}\label{thm.mainmgf}
Suppose $\phi_i - \mu_i$ has cumulant generating function upper bounded by $\psi_i(\lambda)$ over domain $[0,b_i)$ where $0<b_i\leq \infty$. Suppose $\psi_i(\lambda)$ is convex, $\psi_i(0) = \psi_i'(0) = 0$. Define the \emph{expected cumulant generating function} $\bar{\psi}(\lambda)$ as
\begin{align}
\bar{\psi}(\lambda) = \mathbb{E}_T \psi_T(\lambda), \lambda\in [0, \min_i b_i).
\end{align} 
Then,
\begin{align}
\mathbb{E}[\phi_T - \mu_T] \leq (\bar{\psi})^{*-1}(I(T; \bm{\phi})),
\end{align}
where $\bar{\psi}^{*-1}$ is the inverse of the convex conjugate of the function $\bar{\psi}$. 
\end{theorem}

\begin{theorem}\label{thm.pnorm}
Suppose $\phi_i - \mu_i$ has its $\beta$-norm upper bounded by $\sigma_i$, where $1<\beta \leq \infty$. Define $\alpha$, the conjugate of $\beta$ via relation $\frac{1}{\alpha} + \frac{1}{\beta}= 1$. Then,
\begin{align}
|\mathbb{E}[\phi_T - \mu_T] |\leq \| \sigma_T \|_\beta I_\alpha(T;\bm{\phi})^{1/\alpha}. 
\end{align}
Moreover, for $\beta = 2$, we have
\begin{align}
|\mathbb{E}[\phi_T - \mu_T]| \leq \| \sigma_T \|_2 \sqrt{n-1},
\end{align}
and for $2<\beta\leq \infty$ ($1\leq \alpha<2$), we have
\begin{align}
|\mathbb{E}[\phi_T - \mu_T]| & \leq \|\sigma_T \|_\beta (1+n^{\alpha-1})^{1/\alpha} \\
&\le 2^{\frac{1}{\alpha}} \| \sigma_T \|_\beta n^{1/\beta}.
\end{align}
\end{theorem}

\begin{corollary}\cite[Prop. A.1.]{Russo--Zou2015much}\label{cor.gaussian}
Suppose $\phi_i - \mu_i$ is $\sigma_i$-sub-Gaussian. Then, 
\begin{align}
\mathbb{E}[\phi_T - \mu_T] \leq \|\sigma_T \|_2 \sqrt{2I(T; \bm{\phi})}. 
\end{align}
\end{corollary}

\begin{proof}
Applying Theorem~\ref{thm.mainmgf} with $\phi_i(\lambda) = \frac{\lambda^2 \sigma_i^2}{2}$, we have $\bar{\psi}(\lambda) = \frac{\lambda^2 \mathbb{E}_T \sigma_T^2}{2}$. Computing the convex conjugate of $\bar{\psi}(\lambda)$ leads us to
\begin{align}
(\bar{\psi})^{*-1}(x) = \sqrt{2x\mathbb{E}_T \sigma_T^2},
\end{align}
which proves the corollary. 
\end{proof}

\begin{corollary}\label{cor.subgamma}
Suppose $\phi_i - \mu_i$ are sub-Gamma random variables belonging to $\Gamma_+(\sigma^2,c)$. Then, 
\begin{align}
\mathbb{E}[\phi_T - \mu_T] \leq \sigma \sqrt{2I(T; \bm{\phi})} + c I(T;\bm{\phi}). 
\end{align}
\end{corollary}

\begin{corollary}\label{cor.subexponential}
Suppose $\phi_i - \mu_i$ is sub-exponential with parameter $(\sigma,b)$. Then,
\begin{align}
\mathbb{E}_T[\phi_T - \mu_T] \leq \begin{cases} \sigma\sqrt{2I(T; \bm{\phi})} & \textrm{if }I(T;\bm{\phi}) \leq \frac{\sigma^2}{2b} \\ bI(T;\bm{\phi}) + \frac{\sigma^2}{2b^2} & \textrm{otherwise} \end{cases}.
\end{align}
\end{corollary}

Note that Corollary~\ref{cor.subexponential} is a strengthened version of~\cite[Prop. A.2.]{Russo--Zou2015much}. 

\section{The path from Theorem~\ref{thm.russozou} to our main results}\label{sec.generalizationpath}

Underlying the proof of~\cite[Prop. 3.1.]{Russo--Zou2015much} is the Donsker--Varadhan variational representation of relative entropy stated below, and the data processing property of KL divergence. 

\begin{lemma}[Donsker--Varadhan]\label{lemma.donsker}
Let $P,Q$ be probability measures on $\mathcal{X}$ and let $\mathcal{C}$ denote the set of functions $f: \mathcal{X}\mapsto \mathbb{R}$ such that $\mathbb{E}_Q[e^{f(X)}]<\infty$. If $D(P\|Q) <\infty$ then for every $f\in \mathcal{C}$ the expectation $\mathbb{E}_P[f(X)]$ exists and furthermore
\begin{align}
D_{\mathrm{KL}}(P\|Q) = \sup_{f\in \mathcal{C}} \mathbb{E}_P[f(X)] - \ln \mathbb{E}_Q[e^{f(X)}],
\end{align}
where the supremum is attained when $f = \ln \frac{dP}{dQ}$. 
\end{lemma}

It is clear that the application of Lemma~\ref{lemma.donsker} relies on existence of the cumulant generating function but not sub-Gaussianness. It leads to the following proof of Theorem~\ref{thm.mainmgf}. Theorem~\ref{thm.mainmgf} can also be viewed as an application of the transportation lemma~\cite[Lemma 4.18]{Boucheron--Lugosi--Massart2013}. 

\begin{proof}
(of Theorem~\ref{thm.mainmgf}) Applying Lemma~\ref{lemma.donsker} and setting $P = P_{\phi_i | T = i}, Q = P_{\phi_i}, f = \lambda (\phi_i - \mu_i), \lambda >0$, we have
\begin{align}
\lambda(\mathbb{E}[\phi_i|T = i] - \mu_i) & = \mathbb{E}_P f \\
& \leq \ln \mathbb{E}_Q e^f + D_{\mathrm{KL}}(P_{\phi_i | T = i} \| P_{\phi_i}) \\
& \leq \psi_i(\lambda) + D_{\mathrm{KL}}(P_{\bm{\phi}|T = i} \| P_{\bm{\phi}}),
\end{align}
where in the last step we have used the fact that the cumulant generating function of $\phi_i - \mu_i$ is upper bounded by $\psi_i(\lambda)$, and the data processing inequality for the relative entropy. 

Taking expectation with respect to $T$ on both sides, we have
\begin{align}
\lambda \mathbb{E}(\phi_T - \mu_T) \leq \bar{\psi}(\lambda) + I(T; \bm{\phi}),
\end{align}
which implies
\begin{align}
\mathbb{E}(\phi_T - \mu_T) & \leq \inf_{\lambda \in [0, \min_i b_i)} \frac{\bar{\psi}(\lambda) + I(T;\bm{\phi})}{\lambda} \\
& = (\bar{\psi})^{*-1}(I(T;\bm{\phi})),
\end{align}
where in the last step we have used~\cite[Lemma 2.4]{Boucheron--Lugosi--Massart2013}. 
\end{proof}

It is interesting to consider how one can generalize Theorem~\ref{thm.mainmgf} to distributions whose moment generating functions do not exist. Intuitively, a natural 
generalization of Lemma~\ref{lemma.donsker} would lead to generalizations of Theorem~\ref{thm.mainmgf}. In particular, the generalization of Lemma~\ref{lemma.donsker} should involve some $\phi$-divergence, since $\phi$-divergences are the only decomposable divergences that satisfy the data processing inequality for alphabet at least three~\cite{jiao2014information}, and the data processing property is needed in the proof of Theorem~\ref{thm.mainmgf}. 

The literature consists of two generalization paths from the Donsker--Varadhan theorem: one is to go through the Fenchel--Young inequality in convex duality theory, and the other is to go through H\"{o}lder's inequality. It is intriguing that both generalizations lead to the same results presented in Theorem~\ref{thm.pnorm}.  

\subsection{Generalization through H\"{o}lder's inequality}

We first present the generalization path through H\"{o}lder's inequality investigated in~\cite{Atar2015robust}. Note that $\mathbb{E}_P f = \lim_{\alpha \to 0^+} \frac{1}{\alpha} \ln \int e^{\alpha f} dP$. Applying H\"{o}lder's inequality, we have
\begin{align}
\int e^{\alpha f}dP & = \mathbb{E}_Q e^{\alpha f} \frac{dP}{dQ} \\
& \leq (\mathbb{E}_Q e^{\alpha \beta f})^{1/\beta} (\mathbb{E}_Q (\frac{dP}{dQ})^\gamma)^{1/\gamma},
\end{align}
where $\frac{1}{\beta} + \frac{1}{\gamma} = 1, \beta >1, \gamma>1$. Similar arguments were also used in~\cite{courtade2014cumulant}. 

Defining $c = \alpha \beta >\alpha$, rearranging terms, taking logarithm and dividing both sides by $\alpha$, we have
\begin{align}\label{eqn.atar}
\frac{1}{\alpha} \ln \int e^{\alpha f} dP \leq \frac{1}{c}\ln E_Q e^{cf} + \frac{c-\alpha}{c\alpha} \ln \mathbb{E}_Q \left( \frac{dP}{dQ} \right)^{\frac{c}{c-\alpha}}. 
\end{align}

It is clear that (\ref{eqn.atar}) is a generalization of Lemma~\ref{lemma.donsker}. Indeed, taking $\alpha \to 0^+, c = 1$, we have
\begin{align}
\mathbb{E}_P f \leq \ln E_Q e^f + D(P\|Q). 
\end{align}

Now we present a proof of Theorem~\ref{thm.pnorm} using H\"{o}lder's inequality. 
\begin{proof}
(of Theorem~\ref{thm.pnorm}) Setting $P = P_{\phi_i | T =i}, Q= P_{\phi_i}, \Delta_i = \phi_i - \mu_i$ and noting that $\mathbb{E}_Q \Delta_i = 0$, it follows from H\"{o}lder's inequality that
\begin{align}
\left| \mathbb{E}_P \Delta_i \right| & = \left|\mathbb{E}_Q  \Delta_i \frac{dP}{dQ}\right|  \\
& = \left|\mathbb{E}_Q \Delta_i \left( \frac{dP}{dQ} -1 \right) \right| \\
& \leq \left( \mathbb{E}_Q |\Delta_i|^\beta \right)^{1/\beta} \left( \mathbb{E}_Q \left| \frac{dP}{dQ}-1\right|^\alpha \right)^{1/\alpha},
\end{align}
which implies
\begin{align}
\left| \mathbb{E}[\phi_i|T = i] - \mu_i\right | & \leq \sigma_i D^{1/\alpha}_{\phi_\alpha}(P_{\phi_i|T = i} \| P_{\phi_i}) \\
& \leq \sigma_i D^{1/\alpha}_{\phi_\alpha}(P_{\bm{\phi}|T = i} \| P_{\bm{\phi}}),  \label{eqn.lambdapnorm}
\end{align}
where in the last step we used the data processing inequality of the $\phi_\alpha$-divergence. 

If $\beta = \infty$, we have $\alpha=1$ and
\begin{align}
\left| \mathbb{E}[\phi_i|T = i] - \mu_i \right | & \leq \max_i \sigma_i D_{\phi_1}(P_{\bm{\phi}|T = i} \| P_{\bm{\phi}}). 
\end{align}
Taking expectations with respect to $T$ on both sides, we have
\begin{align}
\left| \mathbb{E}[\phi_T - \mu_T] \right | & \leq \max_i \sigma_i I_1(T; \bm{\phi}) \\
& = \| \sigma_T \|_\infty I_1(T; \bm{\phi}). 
\end{align}

If $1\leq \beta <\infty$, applying Young's inequality to (\ref{eqn.lambdapnorm}), we have
\begin{align}
\left | \mathbb{E}[\phi_i | T = i] - \mu_i \right | & \le \inf_{\lambda>0} \frac{1}{\lambda}\left[\frac{\lambda^\beta \sigma_i^\beta}{\beta} + \frac{D_{\phi_\alpha}(P_{\bm{\phi}|T = i} \| P_{\bm{\phi}})}{\alpha}\right].
\end{align}
Taking expectations on both sides with respect to $T$, and using $\bE \inf_\lambda X_\lambda \le \inf_\lambda \bE X_\lambda$, we have
\begin{align}
\left| \mathbb{E}(\phi_T - \mu_T) \right | & \leq \inf_{\lambda>0}\frac{1}{\lambda}\left[ \frac{\lambda^\beta \| \sigma_T \|_\beta^\beta}{\beta} + \frac{I_\alpha(T;\bm{\phi})}{\alpha}\right]\\
& = \| \sigma_T \|_\beta I_\alpha^{1/\alpha}(T;\bm{\phi}). 
\end{align}
The remaining results in Theorem~\ref{thm.pnorm} follow from Lemma~\ref{lemma.boundalphamutual}. 
\end{proof}

\subsection{Generalization through the Fenchel--Young inequality}

We have the following natural variational representation of $\phi$-divergences which is well known in the literature~\cite{nguyen2010estimating}. 
\begin{align}
D_\phi(P\|Q) & = \int \phi\left(\frac{dP}{dQ} \right) dQ \\
& = \int \sup_f \left( \frac{dP}{dQ} f - \phi^*(f) \right)dQ \\
& \geq \mathbb{E}_P f - \mathbb{E}_Q \phi^*(f), \label{eqn.variational}
\end{align}
if $\mathbb{E}_Q \phi^*(f)$ exists. The equality is attained when $f = \phi'\left(\frac{dP}{dQ}\right)$. 

Note that specializing the variational representation above to $\phi(x) = x\ln x - x+1$ cannot directly lead us to the Donsker--Varadhan result. Indeed, after taking $\phi(x) = x\ln x -x +1$, we have
\begin{align}\label{eqn.newdon}
D_{\mathrm{KL}}(P \| Q ) = \sup_f \mathbb{E}_P [f(X)] - \left(  \mathbb{E}_Q [e^{f(X)}] -1 \right),
\end{align}
which is weaker than the Donsker--Varadhan result since $
\ln \mathbb{E}_Q [e^{f(X)}] \leq  \mathbb{E}_Q e^{f(X)}-1$. This phenomenon was already observed in the literature~\cite{ruderman2012tighter}. Indeed, it is because in the KL divergence case $\mathbb{E}_Q e^{f}-1$ is in fact the convex conjugate of the convex function $D(\cdot\|Q)$ when $P$ takes value in the space of all measures, but $\ln \mathbb{E}_Q e^f$ is the convex conjugate of $D(\cdot\|Q)$ when $P$ is constrained to be a probability measure. Indeed, shrinking the primal space would decrease the convex conjugate. Also, it is clear that (\ref{eqn.newdon}) cannot be used to derive Theorem~\ref{thm.mainmgf}. 

However, we can obtain the Donsker--Varadhan result from (\ref{eqn.variational}). Indeed, we have
\begin{align}
D_{\mathrm{KL}}(P \| Q ) & = \sup_f \mathbb{E}_P [f(X)] - \left(  \mathbb{E}_Q [e^{f(X)}] -1 \right) \\
& = \sup_{f+\lambda} \mathbb{E}_P [f + \lambda] - \left( \mathbb{E}_Q[e^{f + \lambda}] -1 \right).
\end{align}

Setting $\mathbb{E}_Q [e^{f(X) + \lambda}] = 1$, we have $\lambda = -\ln \mathbb{E}_Q [e^f(X)]$. It suffices to verify that $f+\lambda = f - \ln\mathbb{E}_Q [e^f(X)] $ can still attain the value $\ln \frac{dP}{dQ}$. Indeed, it is equal to $\ln \frac{dP}{dQ}$ when $f(X) = \ln \frac{dP}{dQ}$. 

Analogously, one obtains the following variational representation of $D_{\phi_\alpha}(P\|Q)$ for $\phi_\alpha(x) = |x-1|^\alpha,\alpha\geq 1$:
\begin{align}
\frac{1}{\alpha} D_{\phi_\alpha}(P\|Q) = \sup_f \mathbb{E}_P[f] - E_Q[f] - \mathbb{E}_Q \frac{|f|^\beta}{\beta},
\end{align}
where $\frac{1}{\alpha} + \frac{1}{\beta} = 1$. Following a path similar to that in the proof of Theorem~\ref{thm.mainmgf}, by setting $P=P_{\phi_i | T = i}, Q = P_{\phi_i}, f = \lambda (\phi_i - \mu_i), \lambda>0$, one arrives at the results of Theorem~\ref{thm.pnorm}. 

%

\section{Tightness of the bounds}\label{sec.tightness}

Theorem~\ref{thm.pnorm} essentially shows that if all the $\phi_i - \mu_i$ have $\beta$-norm bounded, then the exploration bias is upper bounded by $n^{1/\beta}$ of the $\beta$-norm if $2\leq \beta <\infty$. We now show through extreme value theory that it is essentially tight for certain heavy tailed distributions.

Suppose all the $\phi_i, 1\leq i\leq n$ are i.i.d. with CDF
\begin{align}\label{eqn.cdfheavytail}
F(x) = 1- \frac{x_0^\beta (\ln x_0)^c}{x^\beta (\ln x)^c},\quad x\geq x_0 > e^{1/\beta},
\end{align}
where $c>1$. 

Defining $T = \argmax_i \phi_i$, we now care about the asymptotic distribution of $\phi_T$ as $n\to \infty$. The Type II (Fr\'echet type) extreme value distribution with parameter $\beta>0$ is characterized by 
\begin{align}
\Phi_\beta(x) = \begin{cases} 0 & x <0 \\ e^{-x^{-\beta}} & x\geq 0 \end{cases}
\end{align}


\begin{lemma}\cite[Thm. 1.2.1.]{de2007extreme}
 \label{lemma.extremevalue}
The necessary and sufficient condition for a distribution function $F(x)$ to fall into the domain of attraction of $\Phi_\beta(x), 0<\beta<\infty$ is
\begin{align}
\sup\{x: F(x)<1\} = \infty, \textrm{ and } \lim_{t \to \infty} \frac{1-F(tx)}{1-F(t)} = x^{-\beta}, x >0. 
\end{align}
\end{lemma}

Then, it follows from \cite[Cor. 1.2.4.]{de2007extreme} that there exists a sequence $a_n \to \infty$ such that
\begin{align}
\lim_{n\to \infty}\mathbb{P}\left(\frac{\phi_T}{a_n}\leq x\right) = \Phi_\beta(x),
\end{align}
where $a_n = F^{\leftarrow}(1-n^{-1})$ and $F^{\leftarrow}(\cdot)$ denotes the inverse function of $F(x)$. For the CDF in (\ref{eqn.cdfheavytail}) it is easy to see that $a_n$ is the solution to the equation
\begin{align}
x^\beta (\ln x)^c = n (x_0)^\beta (\ln x_0)^c,
\end{align}
which satisfies $a_n \gtrsim \frac{n^{1/\beta}}{(\ln n)^{c/\beta}}$. 

Now we compute the $\beta$-norm of $\phi_i$. We have
\begin{align}
\mathbb{E} \phi_i^\beta & = \int_{0}^\infty \beta x^{\beta -1} \mathbb{P}(\phi_i \geq x) dx \\
& = \int_0^{x_0} \beta x^{\beta -1} dx + \int_{x_0}^\infty \beta x^{\beta -1} \frac{x_0^\beta (\ln x_0)^c}{x^\beta (\ln x)^c}dx \\
& <\infty. 
\end{align}

Hence, for $2\leq \beta < \infty$, Theorem~\ref{thm.pnorm} implies that $\mathbb{E}[\phi_T - \mu_T] \lesssim n^{1/\beta}$. At the same time, it follows from~\cite[Thm. 5.3.1.]{de2007extreme} that 
\begin{align}
\lim_{n\to \infty}\mathbb{E}\left| \frac{\phi_T}{a_n} \right| = \int x d\Phi_\beta(x) = \Gamma\left(1-\frac{1}{\beta}\right),
\end{align}
which shows that $\mathbb{E}\phi_T$ is of order $a_n$, which is at least $\frac{n^{1/\beta}}{(\ln n)^{c/\beta}}$. This shows that the bounds in Theorem~\ref{thm.pnorm} are essentially tight. 

\section{Discussions}\label{sec.discussions}

\subsection{``Soft'' generalizations of the ``hard'' results}

Theorem~\ref{thm.mainmgf} can be viewed as a soft generalization of the following well-known arguments~\cite{Boucheron--Lugosi--Massart2013} through replacing $\ln n$ with $I(T;\bm{\phi})$. Suppose we have $n$ random variables $Z_i$ such that $\mathbb{E}Z_i = 0$, and the moment generating function of $Z_i$ is upper bounded by $\psi_i(\lambda), \lambda \in [0,b)$. Assume that $\psi_i$ is convex, $\psi_i(0) = \psi_i'(0) = 0$. 
\begin{align}
e^{\lambda \mathbb{E}\max Z_i} & \leq \mathbb{E} e^{\lambda \max Z_i} = \mathbb{E} \max e^{\lambda Z_i}  \leq \sum_{i =1}^n \mathbb{E} e^{\lambda Z_i} \\
& \leq n e^{\max_i\psi(\lambda)}. 
\end{align}

Taking logarithm, we have
\begin{align}
\mathbb{E}\max Z_i & \leq \inf_{\lambda \in (0,b)} \left( \frac{\ln n + \max_i\psi(\lambda)}{\lambda} \right) \\
& = (\max_i\psi)^{*-1}(\ln n),
\end{align}
where in the last step we have used~\cite[Lemma 2.4]{Boucheron--Lugosi--Massart2013}. 

Analogously, Theorem~\ref{thm.pnorm} can be viewed as the generalization of the following argument. For $\beta \geq 1$,
\begin{align*}
(\mathbb{E}\max |Z_i|)^\beta & \leq \mathbb{E} \max |Z_i|^\beta  \leq \sum_{i = 1}^n \mathbb{E}|Z_i|^\beta \leq n \max_{i\leq n} \mathbb{E}|Z_i|^\beta.
\end{align*}
It follows that
\begin{align}
\mathbb{E}\max |Z_i| \leq n^{1/\beta} \cdot \max_{i\leq n} \| Z_i \|_\beta. 
\end{align}
However, we note that Theorem~\ref{thm.pnorm} is not a perfect generalization of this ``hard'' argument. For example, we are only able to bound the RHS of Theorem~\ref{thm.pnorm} uniformly over the distribution of $T$ when $\beta \geq 2$, but the ``hard'' argument presented above applies equally to all $\beta \geq 1$. 

More generally~\cite{pollardonline2016}, if $\psi$ is a nonnegative, convex, strictly increasing function on $\mathbb{R}_+$ that satisfies $\psi(0) = 0$, then, for each $\sigma>0$,
\begin{align}
\psi\left( \mathbb{E} \max_{i\leq n} \frac{|Z_i|}{\sigma} \right) & \leq \mathbb{E} \max_{i\leq n} \psi\left(\frac{|Z_i|}{\sigma} \right) \leq \sum_{i\leq n} \mathbb{E} \psi\left(\frac{|Z_i|}{\sigma} \right) \\
& \leq n \max_{i\leq n} \mathbb{E}\psi\left(\frac{|Z_i|}{\sigma} \right).
\end{align}
If $\sigma$ is such that $\mathbb{E}\psi(|Z_i|/\sigma)\leq 1$ for all $i\leq n$, then we have
\begin{align}
\mathbb{E}\max_{i\leq n}|Z_i| \leq \sigma \psi^{-1}(n). 
\end{align}

We note that the generalization of H\"{o}lder's inequality in Orlicz spaces could provide a ``soft'' generalization of the arguments above. For a general \emph{Orlicz} function $\psi: [0,\infty)\mapsto [0,\infty]$, i.e., a convex function vanishing at zero and is also not identically $0$ or $\infty$ on $(0,\infty)$, defining the Luxemburg norm of a random variable $X$ as 
\begin{align}
\| X \|_\psi = \inf \{ \sigma >0: \mathbb{E} \psi \left( \frac{|X|}{\sigma} \right) \leq 1\},
\end{align}
and the Ameniya norm of a random variable $X$ as 
\begin{align}
\| X \|_\psi^A = \inf \left \{ \frac{1 + \mathbb{E} \psi(|tX|)}{t}: t>0 \right \},
\end{align}
we have the generalized H\"{o}lder's inequality:
\begin{lemma}[Generalized H\"{o}lder's Inequality]\cite{hudzik2000amemiya}
Denote an Orlicz function by $\psi$ and its convex conjugate by $\psi^* = \sup\{uv - \psi(v): v\geq 0\}$. Then, 
\begin{align}
\mathbb{E}[XY] \leq \| X \|_\psi \| Y \|_{\psi^*}^A. 
\end{align}
\end{lemma}

The following theorem applies to random variables whose Luxemburg norms are bounded. 
\begin{theorem}\label{thm.orlicznorm}
Suppose $\phi_i - \mu_i$ has its Luxemburg norm upper bounded by $\sigma$. Then,
\begin{align}
|\mathbb{E}[\phi_T - \mu_T]| \leq \sigma \left \| \frac{dP_{T,\bm{\phi}}}{dP_T dP_{\bm{\phi}}} -1 \right \|_{\psi^*}^A,
\end{align}
where $\frac{dP_{T,\bm{\phi}}}{dP_T dP_{\bm{\phi}}}$ follows the product distribution $P_T P_{\bm{\phi}}$ in the Ameniya norm. 
\end{theorem}
\begin{proof}
Denoting $P = P_{\phi_i | T = i}, Q = P_{\phi_i}, \Delta_i = (\phi_i - \mu_i)/\sigma, t>0$, we have
\begin{align}
\left| \mathbb{E}_P [\Delta_i] \right| & = \left| \mathbb{E}_Q \left[ \Delta_i \frac{dP}{dQ} \right] \right| \\
& =  \frac{1}{t}\left| \mathbb{E}_Q \left [ \Delta_i t\left( \frac{dP}{dQ}-1\right) \right] \right| \\
& \leq \frac{1}{t} \left( \mathbb{E}_Q \psi\left( |\Delta_i| \right) + \mathbb{E}_Q \psi^* \left( t \left| \frac{dP_{\phi_i|T = i}}{dP_{\phi_i}} -1 \right| \right) \right) \\
& \leq \frac{1}{t} \left( 1 + \mathbb{E}_{P_{\bm{\phi}}} \psi^* \left( t \left| \frac{dP_{\bm{\phi}|T = i}}{dP_{\bm{\phi}}} -1 \right| \right) \right),
\end{align}
where in the last step we have used the data processing property (i.e., the convexity of $\psi^*(|\cdot|)$). 

Taking expectations with respect to $T$ on both sides and taking infimum for $t$, we have
\begin{align}
|\mathbb{E}[\phi_T - \mu_T]| & \leq \sigma \inf \left \{ \frac{1 + \mathbb{E}_{P_T P_{\bm{\phi}}} \psi^* \left (t \left| \frac{dP_{T,\bm{\phi}}}{dP_T dP_{\bm{\phi}}}-1 \right | \right )}{t}: t>0 \right \} \\
& = \sigma \left \| \frac{dP_{T,\bm{\phi}}}{dP_T dP_{\bm{\phi}}} -1 \right \|_{\psi^*}^A,
\end{align}
where $\frac{dP_{T,\bm{\phi}}}{dP_T dP_{\bm{\phi}}}$ follows the product distribution $P_T P_{\bm{\phi}}$. 
\end{proof}

A natural question is: are there more natural ``soft'' generalizations of all the ``hard'' arguments above? 

\subsection{Connections with other generalizations of mutual information}

There exist various generalizations of mutual information in the literature, and we refer the interested readers to~\cite{ziv1973functionals,verdu2015alpha,lapidoth2016two} for references. The dependence measure $I_\alpha(X;Y)$ introduced in this paper seems to have received only scant attention in the existing literature. Some generalizations such as Sibson's mutual information~\cite{sibson1969information} involve minimizing over an auxiliary distribution $Q_Y$, and the dependence measure in~\cite{lapidoth2016two} involves minimizing jointly over $Q_XQ_Y$. Even when power functions are used to define $\phi$-divergences, functions $x^\alpha,\alpha\geq 1$ plus some affine terms were used much more frequently than $|x-1|^\alpha$ except for the case of $\alpha = 1$ (total variation distance) and $\alpha = 2$ ($\chi^2$-distance). For example, the usual definition of R\'enyi divergence involves the function $x^\alpha$ but not $|x-1|^\alpha$. It remains an interesting question whether other generalizations of mutual information could prove useful in bounding the exploration bias. 


\section{Acknowledgement}

We are grateful to James Zou for discussing with us the results in~\cite{Russo--Zou2015much}, which inspired this work. 

\appendix 


%
%

\section*{Proof of Lemma~\ref{lemma.boundalphamutual}}
We first prove a general result regarding the $\phi$-mutual information $I_\phi(X;Y)\triangleq D_\phi(P_{XY}\|P_XP_Y)$ for a general convex function $\phi$. 
\begin{lemma}\label{lem.phi_mutualinfo_upper}
Let $X$ take value in a finite set $\calX$, then for convex $\phi: \bR_{\ge 0}\mapsto \bR$,  
\begin{align}
I_\phi(X;Y)&\le \phi(0)\left(1-\sum_{x\in \calX}[P_X(x)]^2\right) \nonumber\\
&\qquad\qquad + \sum_{x\in\calX} [P_X(x)]^2\phi\left(\frac{1}{P_X(x)}\right).
\end{align}
If $\phi(t)$ is strictly convex at $t=1$, then the upper bound is tight iff $X$ is a deterministic function of $Y$.
\end{lemma}
\begin{proof}[Proof of Lemma \ref{lem.phi_mutualinfo_upper}]
It follows from the generalization of the Gel'fand-Yaglom-Peres theorem for $\phi$-divergences~\cite{Gilardoni2009gel} that it suffices to consider $Y$ being a discrete random variable. Note that
\begin{align}
0 \le \frac{P_{XY}(x,y)}{P_X(x)P_Y(y)} \le \frac{1}{P_X(x)}
\end{align}
then by the convexity of $\phi$ we know that
\begin{align}
\phi\left(\frac{P_{XY}(x,y)}{P_X(x)P_Y(y)}\right)&\le \left(1-\frac{P_{XY}(x,y)}{P_Y(y)}\right)\phi(0) \nonumber\\
&\qquad +\frac{P_{XY}(x,y)}{P_Y(y)} \phi\left(\frac{1}{P_X(x)}\right).
\end{align}
Now summing over $x,y$ in the definition of $I_\phi(X;Y)$ yields the desired result. When $\phi(t)$ is strictly convex at $t=1$ and the equality holds, we have $P_{XY}(x,y)\in \{0,P_Y(y)\}$ for any $x,y$, which means that $X$ is a function of $Y$.
\end{proof}

Now the first statement in Lemma \ref{lemma.boundalphamutual} follows from Lemma \ref{lem.phi_mutualinfo_upper} applied to $\phi_\alpha(t)=|t-1|^\alpha$. For the second statement, for $\alpha\in[1,2]$ we define $f(t)=t^2[(\frac{1}{t}-1)^\alpha-1]=t^{2-\alpha}(1-t)^\alpha-t^2$ on $[0,1]$. We note that $f(t)$ is \emph{not} concave on $[0,1]$: 
\begin{align}
f''(t) &= -(2-\alpha)(\alpha-1)\left(\frac{1}{t}-1\right)^\alpha - 2(2-\alpha)\alpha \left(\frac{1}{t}-1\right)^{\alpha-1}\nonumber\\
&\qquad + \alpha(\alpha-1)\left(\frac{1}{t}-1\right)^{\alpha-2} - 2
\end{align}
satisfies $f''(1)=+\infty$ for $\alpha\in[1,2)$. However, it is straightforward to see that $f''(t)$ is increasing on $[0,1]$, and
\begin{align}
f''(|\calX|^{-1}) &\le \alpha(\alpha-1)(|\calX|-1)^{\alpha-2}-2 \\
&\le \alpha(\alpha-1)-2\le 0.
\end{align}
Hence, if we define 
\begin{align}
g(t) &\triangleq f(t) - \left[\frac{2(|\calX|-1)^\alpha-2}{|\calX|}-\alpha (|\calX|-1)^{\alpha-1}\right](t-\frac{1}{|\calX|}) \nonumber\\
&\qquad - \frac{(|\calX|-1)^\alpha-1}{|\calX|^2},
\end{align}
it is straightforward to see that $g(|\calX|^{-1})=g'(|\calX|^{-1})=0$ and $g''(t)=f''(t)$ on $[0,1]$. In particular, $g''(|\calX|^{-1})\le 0$ and $g''(t)$ is increasing on $[0,1]$. As a result, the maximum of $g(t)$ over $t\in[0,1]$ is attained at $t=|\calX|^{-1}$ or $t=1$, and
\begin{align}
g(1) &= -1-\left[\frac{2(|\calX|-1)^\alpha-2}{|\calX|}-\alpha (|\calX|-1)^{\alpha-1}\right](1-\frac{1}{|\calX|}) \nonumber\\
&\qquad- \frac{(|\calX|-1)^\alpha-1}{|\calX|^2}\\
&= \frac{(\alpha-1)(|\calX|-1)^\alpha - (2-\alpha)(|\calX|-1)^{\alpha+1}-(|\calX|-1)^2}{|\calX|^2}\\
&= \frac{(|\calX|-1)^\alpha}{|\calX|^2}\left[\alpha-1-(2-\alpha)(|\calX|-1)-(|\calX|-1)^{2-\alpha}\right].
\end{align}
Obviously $g(1)=0$ if $|\calX|=1$, and if $|\calX|\ge 2$, we have
\begin{align}
g(1) &\le \frac{(|\calX|-1)^\alpha}{|\calX|^2}\left[\alpha-1-(2-\alpha)\cdot 1-1\right] \\
&= \frac{2(\alpha-2)(|\calX|-1)^\alpha}{|\calX|^2}\le 0.
\end{align}

In summary, we have $g(t) \le \max\{g(|\calX|^{-1}),g(1)\}=0$ for any $t\in[0,1]$, and thus
\begin{align}
& f(P_X(x))\le \frac{(|\calX|-1)^\alpha-1}{|\calX|^2}+\nonumber\\
&\left[\frac{2(|\calX|-1)^\alpha-2}{|\calX|}-\alpha (|\calX|-1)^{\alpha-1}\right](P_X(x)-\frac{1}{|\calX|}) .
\end{align}
Now summing over $x\in\calX$ completes the proof.

\bibliographystyle{IEEEtran}
\bibliography{di}
\end{document}